\RequirePackage{fix-cm}
\documentclass[smallextended]{svjour3}       
\smartqed  
\usepackage{graphicx}
\usepackage{csquotes}
\usepackage[table]{xcolor}
\usepackage[authoryear]{natbib}
\usepackage{amsmath}
\usepackage{amssymb}
\usepackage{enumitem}
\usepackage{bbm}

\begin{document}

\title{Kripke Semantics of the Perfectly Transparent Equilibrium}


\author{Ghislain Fourny}


\institute{G. Fourny \at
              ETH Z\"urich \\
              Department of Computer Science \\
              \email{ghislain.fourny@inf.ethz.ch}\\
}

\date{July 19, 2018}

\maketitle

\begin{abstract}
The Perfectly Transparent Equilibrium is algorithmically defined, for any game in normal form with perfect information and no ties, as the iterated deletion of non-individually-rational strategy profiles until at most one remains. It is always Pareto optimal and thus provides a simple way out of social dilemmas. In this paper, we characterize the Perfectly Transparent Equilibrium with adapted Kripke models having necessary rationality, necessary knowledge of strategies -- necessary epistemic (factual) omniscience -- as well as eventual logical omniscience.

Eventual logical omniscience is introduced as a weaker version of perfect logical omniscience, with logical omniscience being quantized and fading away counterfactually. It is the price to pay for necessary factual omniscience and necessary rationality: we conjecture that epistemic omniscience, logical omniscience and necessary rationality form an impossibility triangle.

We consider multimodal classes of Kripke structures, with respect to agents, but also in the sense that we have both epistemic and logical accessibility relations. Knowledge is defined in terms of the former, while necessity is defined in terms of the latter. Lewisian closest-state functions, which are not restricted to unilateral deviations, model counterfactuals.

We use impossible possible worlds \`a la Rantala to model that some strategy profiles cannot possibly be reached in some situations. Eventual logical omniscience is then bootstrapped with the agents' considering that, at logically possible, but non-normal worlds \`a la Kripke, any world is logically accessible and thus any deviation of strategy is possible. As in known in literature, under rationality and knowledge of strategies, these worlds characterize individual rationality. Then, in normal worlds, higher levels of logical omniscience characterize higher levels of individual rationality, and a high-enough level of logical omniscience characterizes, when it exists, the Perfectly Transparent Equilibrium.

\end{abstract}

\section{Introduction: perfect prediction and epistemic omniscience}

In classical game theory and, more generally, the neoclassical school of economics, we typically build rationality on the top of an assumption that agents can unilaterally change their strategies. With this mindset, \citet{Stalnaker1994} applied \citet{Kripke1963}'s work to game theory and showed that rationalizability is characterized with common belief in rationality, and that Nash equilibria are characterized with rationality and knowledge of of the opponent's beliefs about one's strategy.

Allowing only for unilateral deviations, however, leaves us with the insatisfaction of accounting for -- predicting or describing -- agent's decisions, while simultaneously assuming that they are unpredictable. This leads to tricky situations in which reasonings involve states in which some agents are not rational, or for games in extensive form, situations in which reasoning is made at nodes that are not actually reached if all agents are rational. This is widely discussed in literature, for example by \citet{Aumann1995}, \citet{Binmore1996}, \citet{Binmore1997}, \citet{Stalnaker1998}, \citet{Halpern2001}.

\citet{Dupuy1992} showed that understanding the prediction structure in Newcomb's problem was directly connected to the Prisoner's dilemma. In \citep{Dupuy2000}, he suggested that an alternate account of rationality could be used to suggest new solution concepts in which agents perfectly predict each other's strategies in all possible worlds, and are rational in all possible worlds. This was formalized by \citet{Fourny2018} as the Perfect Prediction Equilibrium for games in extensive form, and later extended \citep{Fourny2017} to games in normal form as the Perfectly Transparent Equilibrium (PTE).

\citet{Halpern:2013aa} characterized a new construct (individually rational miminax-rationalizable outcomes) with Common Counterfactual Belief of Rationality (CCBR), which is a weaker form of our assumption of rationality in all possible worlds (necessary rationality). As it turns out, the PTE often coincides with Halpern's and Pass's solution concept, but some counterexamples are known. Other papers in literature investigate alternate solution concepts that are more transparent than Nash equilibria \citep{Bilo2011} \citep{Shiffrin2009} \citep{Tennenholtz2004}. Their relationship with the PTE is discussed in \citep{Fourny2017}.

This paper aims at providing a characterization in Kripke semantics of the PTE. As a consequence, this provides a clear understanding of where Perfect Prediction and CCBR differ.

We are going to make the outrageously unrealistic assumption that players are epistemically omniscient\footnote{This is an intended hint to \citep{Stalnaker1994}, who makes this statement about logical omniscience}\footnote{In spite of seeming unrealistic, this assumption is very relevant, as Dupuy explains that Perfect Prediction does not need to hold: it suffices that players \emph{believe} that they are all perfect predictors to reach the outcome characterized by Perfect Prediction. Indeed, the commonly predicted outcome is dictated by the laws of logics, and the laws of logics alone, as beautifully explained by \citet{Hofstadter1983}.}. There is in particular knowledge of strategies in all possible worlds, known as necessary knowledge of rationality\footnote{This is a stronger assumption than common knowledge of strategies}. In addition, we assume rationality in all possible worlds, known as necessary rationality\footnote{This is a stronger assumption than common knowledge of rationality}.

As all strategies are known by all agents in all possible worlds, the price to pay for this epistemic omniscience is a weakening of logical omniscience. We consider Kripke models arranged in layers of increasing degrees of logical omniscience, in such a way that logical omniscience gradually degrades counterfactually. Non-normal worlds \citep{Kripke1965} in which anything is logically possible are eventually reached after a finite number of nested counterfactual deviations, which bootstraps what we call eventual logical omniscience and supports formally the iterated deletion of non-individually-rational strategy profiles underlying the Perfectly Transparent Equilibrium.

The remainder of this paper is organized as follows. Section  \ref{section-background} recalls the basic concepts of the PTE, including games in normal form. Section \ref{section-kripke-framework} introduces our adapted Kripke framework to describe our assumption of necessary rationality, necessary knowledge of strategies and eventual logical omniscience in terms of possible worlds and accessibility relations. Section \ref{section-further-tools} introduces classical Kripke semantics concepts such as knowledge and necessity, as well as our conception of rationality under epistemic omniscience, and eventual logical omniscience. Section \ref{section-modal-logic} introduces the syntax for our modal logic constructs (rationality, knowledge of strategies, necessity, etc). Section \ref{section-characterization} states the equivalence between Kripke models with necessary rationality, necessary knowledge of strategy and eventual logical omniscience and the PTE. Section \ref{section-conclusion} wraps up our findings and summarizes differences with CCBR.

\section{Background: games in normal form and the Perfectly Transparent Equilibrium}
\label{section-background}

The Perfectly Transparent Equilibrium adapts \citet{Dupuy2000}'s framework of Perfect Prediction to games in normal forms: strategic games. We are interested in settings with perfect information in that players all know the details of the game. We consider pure strategies in the sense that players do not assemble mixed strategies with weights, and we consider games with no random moves (no dice rolls) and no ties\footnote{also called ``in general position'' in literature.}.

In this section, we summarize the most important definitions around the PTE. More details as well as numerous examples can be found in \citep{Fourny2017}.

\subsection{Games in normal form}

A game in normal form can be described with a matrix mapping each possible combination of the players' choices of strategy, called a strategy profile, to the payoffs they receive.

\begin{definition}[Game in normal form]

A game in normal form $\Gamma=(P, (\Sigma_i)_i, (u_i)_i)$ is defined with:

\begin{itemize}
\item a finite set of players $P$.
\item a set of strategies $\Sigma_i$ for each player $i\in P$. The set of strategy profiles is denoted $\Sigma=\times_{i\in P} \Sigma_i$.
\item a specification of payoffs $u_i(\overrightarrow\sigma)$ for each player $i\in P$ and strategy profile $\overrightarrow\sigma=(\sigma_j)_{j\in P}$.
\end{itemize}

\end{definition}

We only consider pure strategies, meaning that players may not use randomness to build mixed strategies. The outcome of a game must thus always be one of the strategy profiles of the normal form matrix, with each player getting the corresponding payoff.

Furthermore, in this paper, we assume that there are no ties, meaning that a player always has a strict preference between any two strategy profiles.

Payoffs only have an ordinal meaning, which is why in all our examples we use an increasing sequence of small integers starting with 0.

Finally, the game is only played once: it is important to understand that this is not a repeated-game equilibrium.

\subsection{Preemption}

Algorithmically, the Perfectly Transparent Equilibrium is obtained by iterated elimination of non-individually-rational strategy profiles until at most one remains. In particular, when it exists, it is unique.

\subsubsection{The first round of elimination: individual rationality}

We now recall the computation of the PTE, as it was introduced in \citep{Fourny2017}. The core idea is that it generalizes individual rationality to higher levels that consider the fact that some strategy profiles were eliminated on lower levels and cannot be possibly commonly known as the solution.

Level-1 individual rationality coincides with individual rationality as broadly found in literature, e.g., in the folk theorem characterizing stable equilibria in repeated games, and defines a subset of all strategy profiles that Pareto-dominate the tuple of maximin payoffs.

Level-2 individual rationality re-iterates this same idea, keeping strategy profiles that Pareto-dominate the tuple of maximin payoffs but taking into consideration only strategy profiles from level 1, and so on. Strategy profiles get iteratively deleted as one considers higher levels of individual rationality.

At some point, either one profile remains and is stable, which is the PTE, or all profiles get eliminated, and there is no equilibrium.

The following definition summarizes this process\footnote{For brevity, we merged the definition for level-1 individual rationality with the general definition originally given separately in \cite{Fourny2017} for pedagogical purposes, as the former is a special case of the latter if we properly initialize $\mathcal{S}_0$ to the set of all strategy profiles.}.

\begin{definition}[Level-k individually rational strategy profile]
Given a game in normal form $\Gamma=(P, \Sigma, u)$ with no ties, a strategy profile is level-k preempted, for $k\ge1$, if it does not Pareto-dominate the maximin utility, where the maximin is only taking into account strategy profiles that are not level-$(k-1)$-preempted\footnote{That is, that are level-$(k-1)$ individually rational.}. The recursion is initialized with $\mathcal{S}_0(\Gamma)=\Sigma$, that is, all profile are level-0 individually rational and none are level-0 preempted. The strategy profiles $\overrightarrow\sigma\in \mathcal{S}_k(\Gamma)$ that survive the $k^{th}$ round of elimination are called level-k individually rational, and are characterized with:

$$\mathcal{S}_k(\Gamma) = \{ \overrightarrow\sigma | \forall i \in P, u_i(\overrightarrow\sigma) \ge$$

$$\max_{
\begin{array}{c}
\tau_i\in\Sigma_i
\\
{\scriptstyle \text{s.t.} \exists \tau_{-i}\in\Sigma_{-i}, (\tau_i, \tau_{-i})\in \mathcal{S}_{k-1}(\Gamma)}
\end{array}
}
\quad
\min_{
\begin{array}{c}
\tau_{-i}\in\Sigma_{-i}
\\
{\scriptstyle \text{s.t.} (\tau_i, \tau_{-i})\in \mathcal{S}_{k-1}(\Gamma)}
\end{array}
} u_i(\tau_i, \tau_{-i}) \}$$
\end{definition}

\subsection{Perfectly Transparent Equilibrium}

The PTE is defined as the one strategy profile, if any, that survives the above iterated elimination of strategy profiles, i.e., it is the only one that is level-k individually rational for arbitrarily high k.

\begin{definition}[Perfectly Transparent Equilibrium for games in normal form]
Given a game in normal form with no ties, a Perfectly Transparent Equilibrium is a strategy profile that is level-k individually rational for all $k\in \mathbb{N}$. The set of Perfectly Transparent Equilibria is $\mathcal{S}(\Gamma)=\cap_k \mathcal{S}_k(\Gamma)$.
\end{definition}

The PTE is at most unique. Uniqueness of the Perfectly Transparent Equilibrium is crucial to the justification of Perfect Prediction, as the players have no magic powers. Rather, they predict the outcome of the game thanks to the laws of logic and their awareness that all players are rational, reason perfectly well, and because there is common knowledge thereof in all possible worlds.

Since the laws of logics dictate at most one answer, as explained by \citet{Hofstadter1983}, the fact that they all arrive to the same conclusion is not due to telepathic powers, but to their sole mastery of modal logic.

It is this modal logic that we are now going to give a formal account of, characterizing the Kripke models underlying the Perfectly Transparent Equillibrium.

\section{Kripke framework}
\label{section-kripke-framework}

We now introduce our adapted Kripke models to characterize the PTE in terms of modal logic.

\subsection{The general idea}

\citet{Stalnaker1994} introduced the idea of using \citet{Kripke1963}'s seeding work to characterize and evaluate solution concepts. A solution concept can thus be characterized algorithmically (i.e., for rationalizability, iterated deletion of dominated strategies) as well as epistemically (i.e., common belief in rationality).

The algorithmic characterization of the PTE was given by \citet{Fourny2017} and recalled in Section \ref{section-background}. In this paper, we give a characterization of the PTE in terms of \citet{Kripke1963} semantics. 

In Stalnaker's original work, and in most of the field of classical game theory, only unilateral deviations or strategies are considered, making it a purely epistemic approach as described for example by \citet{Perea2012}. However, Perfect Prediction introduces rationality in all possible worlds as well as common knowledge of strategies in all possible worlds, a machinery that only makes sense by dropping the assumption of unilateral deviations: counterfactuals, in the way that they were introduced by \citet{Lewis1973} are thus needed to express the consequences of a deviation of strategies. For example, in the prisoner's dilemma, we can express statements justifying rational cooperation such as ``If I were to defect, the opponent would defect as well.'' This is implemented with closest state functions.

A related approach introducing counterfactuals after Lewis and dropping the assumption of unilateral deviations was contributed by \cite{Halpern:2013aa} for characterizing Common Counterfactual Belief of Rationality (CCBR). Necessary rationality and necessary knowledge of strategies are stronger and more radical assumptions than CCBR, although we will use similar closest-state functions as in \cite{Halpern:2013aa}. However, unlike Halpern and Pass, we do not have any probability distributions modeling beliefs, as we are working with epistemically omniscient agents (Perfect Predictors).

We use multimodal Kripke models, not only in terms of agents, but also in terms of kind of accessibility relation. We explicitly distinguish between epistemic accessibility to model what worlds agents consider epistemically possible, and logical accessibility to model worlds that agents consider logically possible (possibly counterfactually), and include them both in our Kripke models.

We introduce impossible possible worlds, both in a way similar to \citet{Kripke1965}, with non-normal worlds in which everything is possible and nothing is necessary, and also impossible possible worlds similar to \citep{Rantala1982}, which are logically impossible worlds in which we set truth values manually.

In any possible world, we will have common knowledge of rationality, and common knowledge of strategies\footnote{A strategy is not an event, so that the term ``common knowledge of strategies'' may appear as a surprise. We will formally explain what we mean with ``knowledge of strategies'' by extending the concept of knowledge to any ``world variable.''}. As this is true in all possible worlds, it follows that we have a much more stringent assumption than in Nashian game theory, namely: \emph{necessary rationality\footnote{The closest equivalent to this found in literature is \citep{Halpern:2013aa}, who introduce Common Counterfactual Knowledge of Rationality. Our approach is more stringent, as this also entails that an agent would believe to also have been counterfactually rational (CB*RAT), which is something Halpern and Pass exclude.} and necessary knowledge of strategies}.

The price to pay is that we cannot have full logical omniscience in all possible worlds. We thus introduce the notion of level-k logical omniscience, which can be seen as reasoning steps in a ``Turing machine'' that computes the PTE. The level of logical omniscience decreases as agents navigate conterfactuals, until non-normal Kripkean worlds are reached at level 1. The PTE is characterized with Kripke frames fulfilling the above assumptions on high enough levels of logical omniscience.

From a terminology perspective, we will use the words ``Kripke frame'', which is, simply put, a skeleton graph of worlds, as well as ``Kripke structure'' -- a Kripke frame with a labelling function assigning truth values to logical formulas -- and ``Kripke model'' -- a Kripke structure with an actual world.

Let us start by formally defining our Kripke frames.

\begin{definition}[Kripke frame]
A Kripke frame $M$ appropriate for a strategic game $\Gamma=(P, (\Sigma_i)_{i\in P}, (u_i)_{i\in P})$ is a tuple

$$M=(\Omega, \Lambda, \Xi, \overrightarrow{\sigma}, \mathcal{K}, \mathcal{L}, f)$$ where:

\begin{itemize}
\item $\Omega$ is a set of possible worlds (including impossible possible worlds);
\item $\Lambda \subseteq \Omega$ is the set of all logically possible worlds (which can be normal or not);
\item $\Xi \subseteq \Lambda$ is the set of all normal worlds;
\item $\overrightarrow{\sigma}\in \Omega \to \Sigma$ maps these worlds to strategy profiles in $\Sigma=\times_{i\in P} \Sigma_i$;
\item $\mathcal{K}_i\subseteq\Omega^2$ is the epistemic accessibility relation for each player $i\in P$;
\item $\mathcal{L}_i\subseteq\Omega^2$ is the logical accessibility relation for each player $i\in P$;
\item $f\in \Omega \times P \times \Sigma \to \Omega$ is the Lewisian closest-state function used for counterfactual statements;
\end{itemize}
\end{definition}

The following subsections detail the components of this Kripke frame.

\subsection{Set of all possible worlds}

$\Omega$ is the set of all possible worlds. We will consider individual worlds $w \in \Omega$. These worlds may be logically possible or not, and if they are logically possible, they may be normal or not.

In each possible world, the game is played in some specific way, leading to a specific outcome.

\subsection{Logically possible worlds}

$\Lambda\subseteq\Omega$ is the set of all logically possible worlds. These worlds are consistent under the laws of propositional logic (modus ponens, weak consistency, etc) and propositions are assigned truths as per their semantics. These worlds may still be normal or non-normal (see Section \ref{section-normal}).

We handle logically impossible worlds with Rantala's ``impossible possible worlds'' approach, meaning that, in these worlds, logical statements are assigned truth values individually \citep{Rantala1982}. We follow the fourth approach described in \cite{Halpern2011} in Section 2.5 in order to deal with logical omniscience.

As we will see in Section \ref{section-logical-accessibility}, in any world, the agents may have their own subjective assessments on which worlds are logically possible. $\Lambda$ provides the objective basis for what worlds are really logically possible. All agents agree that any world in $\Lambda$ is logically possible, even though they may believe that further worlds are logically possible as well. This is the case in non-normal Kripkean worlds where everything is logically possible and nothing is logically necessary.

Lack of logical omniscience will be modelled with the notion of level-k logical omniscience. The higher the value of $k$, the more logically omniscient the agents.

\subsection{Normal worlds}
\label{section-normal}

$\Xi\subseteq\Lambda$ is the set of all normal worlds. In these worlds, not only the laws of propositional logic apply, but necessity and possibility are also defined in terms of Kripke semantics based on the logical accessibility relation and objective logical possibility.

We handle non-normal worlds with Kripke's ``impossible possible worlds'' approach \citep{Kripke1965}, in the sense that modal operators are defined with nothing being necessary, and everything being possible, rather than according to classical normal semantics.

Non-normal worlds may also be logically impossible, in which case all truth assignments, not only for the modal operators, are set manually.

A normal world is always logically possible.

\subsection{Epistemic accessibility relation}

We have an epistemic accessibility relation $w \mathcal{K}_i w'$ that expresses that in world $w$, agent $i$ considers world $w'$ epistemically possible, in the sense that $w'$ is compatible with their knowledge and the facts that they are aware of.

This epistemic accessibility relation entails the concept of knowledge: an event true in all epistemically accessible worlds from $w$ is known in $w$.

We assume, as is widespread in game-theoretical settings, that the epistemic accessibility relation is reflexive (and thus serial), transitive and symmetric (and thus euclidian), i.e., it organizes the $\Omega$ space into partitions. In the framework of the PTE and Perfect Prediction, epistemic accessibility relations play a lesser role. Indeed, we are considering a class of Kripke models in which the agents are epistemically omniscient on the topics that interest us: each event relevant to us (rationality, etc), is common knowledge in the worlds in which it happens. As we will see, the assumptions we will make (necessary rationality, necessary knowledge of strategies) are stronger than common knowledge.

In practice, the Kripke frames we will construct will use plain and simple equality as the epistemic accessibility relation, even though this does not have to be the case in wider contexts.

\subsection{Logical accessibility function}
\label{section-logical-accessibility}

We have a logical accessibility relation $w \mathcal{L}_i w'$ that expresses that in world $w$, agent $i$ considers world $w'$ to be logically possible.

This alternate accessibility relation implements the approach contributed by \cite{Halpern2011}, in the sense that an agent has a test function that stamps every world in $\Omega$ as being logically possible from their perspective, or not. This test function, in the world $w$, is implemented as $T(w')=w \mathcal{L}_i w'$.

We are considering logical accessibility relation that are total on $\Lambda$, meaning that $\mathcal{L}_i \cap \Lambda^2 = \Lambda^2$. This condition on logical accessibility, which commonly holds in literature in the absence of non-normal worlds, is thus even stronger than having partitions, as there is only one partition in $\Lambda$: $\Lambda$ itself. Any logically possible world is thus logically accessible from any other logically possible world.

For non-normal worlds, that is, in $\Omega \setminus \Xi$, everything is logically possible. Non-normal worlds are thus connected to all worlds in $\Omega$ by the logical accessibility relation.

In normal worlds, all logically possible worlds ($\Lambda$), and only them, are logically accessible. This is all independent of the agent considered, as they reason perfectly.

Note that, for a world that is non-normal but logically possible, both of the above conditions are mutually compatible, i.e., as a logically possible worlds, it is connected to all logically possible worlds, and as a non-normal world, it is connected as well to all logically impossible worlds.

This approach allows us to deal with (eventual) logical omniscience; indeed, we need to bootstrap the reasoning from a state in which nothing is known about logically impossibility, to states in which the agent knows exactly which worlds are logically possible.

Note that we may have $w \mathcal{L}_i w'$ even if $w \mathcal{K}_i w'$ does not hold. In other words, an agent may believe that a world, even if it is incompatible with their knowledge, is logically possible and would have been (counterfactually) logically possible. This corresponds to the seeding work of \cite{Lewis1973} that defines the semantics on counterfactual implications.

\subsection{Strategy profile in each world}

As is done commonly in literature, also in absence of knowledge of strategies, we associate each world $w\in \Omega$ with the strategy profile reached by the players in this world $\overrightarrow{\sigma}(w)$, i.e., $\overrightarrow{\sigma}$ is a function from $\Omega$ to $\Sigma$. We denote $\sigma_i(w)$ the strategy picked by agent $i$ in world $w$.

We also assign strategy profiles to logically impossible worlds and non-normal worlds.

\subsection{Counterfactuals}

Following \cite{Lewis1973}, we model counterfactual implications with a function $f$ associating a world $w\in \Omega$, a player $i\in P$ and a strategy $s'_i \in \Sigma_i$ with the closest possible world $f(w, i, s'_i)\in\Omega$ in which player $i$ picks strategy $s'_i$.

This function $f$ is referred to as the closest-state function.

Note that this world may be logically impossible, or it may be non-normal. However, it always holds that $\sigma_i(f(w, i, s'_i))=s'_i$.

We slightly modify Lewisian semantics in the sense that a logically impossible world may be closer than a logically possible world. Indeed, worlds may show no logical contradictions for a given deviation of strategy in the presence of an agent with a lesser level of logical omniscience, but an agent with a higher level of logical omniscience may have sufficient reasoning skills to conclude that that same deviation of strategy is logically impossible. In such a case, the closest-state function must return an impossible world to the latter deviating agent.

In other words, the spheres of accessibility underlying $f$ are arranged around a possible world in decreasing levels of logical omniscience starting with the actual world at the center. Given an agent $i$ and deviation of strategy $s'_i$, a world with the highest (but lower or equal to that of the actual world) level of logical omniscience and in which agent $i$ picks $s'_i$ is thus chosen as $f(w, i, s'_i)$, whether it is logically accessible or not\footnote{When $f(w, i, s'_i)$ is not logically accessible, it means that $s'_i$ is preempted, and agent $i$ considers this deviation to be logically impossible at his level of logical omniscience.}. At equal levels, a logically possible worlds is closer than a logically impossible world.

\section{Concepts used in Kripke frames}

In this section, we formally introduce a few concepts such as events, world variables, knowledge, necessity, epistemic omniscience, eventual logical omniscience and rationality. Most of these concepts are commonly found in literature. World variables help us formalize knowledge of strategies cleanly as strategies are not events.

\label{section-further-tools}

\subsection{Events and world variables}

We use the definition of an event (also called proposition) commonly found in literature.

\begin{definition}[Event]
An event $E$ is a set of possible worlds:
$$E\subset\Omega$$
\end{definition}

We also introduce the definition of a world variable, having the strategy profile world variable in mind. This definition is similar to that of random variables in probability theory.

\begin{definition}[World variable]
A world variable $X$ is a function that maps each world to a target set $S$:

$$X\subset \Omega \to S$$
\end{definition}

An example of world variable is $\overrightarrow{\sigma}$, which maps each world to a strategy profile.

\subsection{Knowledge}

Knowledge\footnote{Our approach is epistemic rather than doxastic, because we assume reflexivity of the epistemic accessibility relation.} of an event follows the common approach in Kripke semantics, based on the epistemic accessibility relation:

\begin{definition}[Knowledge of an event]
Agent $i\in P$ knows event $E\subset \Omega$ in world $w\in \Omega$ if

$$\forall w' \in \Omega, w\mathcal{K}_i w' \implies w'\in E$$

\end{definition}

We extend this definition of knowledge to knowledge of a world variable like so:

\begin{definition}[Knowledge of a world variable]
A world variable $X \in \Omega \to V$ is known to an agent $i\in P$ in a world $w\in \Omega$ if, in all epistemically accessible worlds from $w$, $X(w)$ has the same value:

$$\forall w' \in \Omega, w\mathcal{K}_i w' \implies X(w') = X(w)$$

\end{definition}

Note that knowledge of an event in a world can seen as the conjunction of the knowledge of its indicator function with the event happening in that world.

In the following sections, we will in particular use knowledge of rationality (an event) and knowledge of strategies (a world variable).

We do not need to introduce common knowledge in our framework, as our assumptions, formulated in terms of necessity, are stronger than common knowledge or even common counterfactual belief.

\subsection{Necessity}

Necessity is to the logical accessibility relation what knowledge is to the epistemic accessibility relation.

\begin{definition}[Necessity of an event]
In a normal world $w\in\Xi$, necessity of an event $E$ holds for agent $i$ if that event includes all logically accessible worlds.

$$\forall w' \in \Omega, w\mathcal{L}_i w' \implies w'\in E$$

In a non-normal world $w\notin\Xi$, no event is ever necessary.
\end{definition}

Note that, because logical worlds are always logically accessible, necessity of an event in any world entails that it is a superset of the $\Gamma$ set of logically possible worlds. Since, in the class of frames considered, epistemic accessibility implies logical accessibility and the latter is total on $\Lambda$, then a necessary event is en event that happens in all worlds in $\Lambda$, and is thus commonly known in all worlds in $\Lambda$. This fact is independent on the agent considered.

Also, note that the necessity of an event is never itself necessary in the presence of non-normal worlds ($\Lambda \setminus \Xi \neq \emptyset$), which is the case in our models\footnote{We thus do not have $\Box\Box RAT$ nor $\Box\Box K(\overrightarrow \sigma)$}.

In particular, we will use necessary rationality and necessary knowledge of strategies\footnote{We could also use a weaker definition of necessity that propagates via the Lewisian closest-state function, similar to Halpern's and Pass's CCBR: an event $E$ is necessary in a world $w$ for agent $i$ if it contains all worlds logically accessible by $i$ from $w$, and, recursively, if would also have been necessary in case of a logically possible deviation of strategies. A necessary event then does not cover the entire $\Lambda$ space, but only the transitive closure of the state graph induced by $(\mathcal{L}_i)_i$ and $f$. We believe that the results in this paper would still hold. However, the simple definition used here, is a more faithful account of Dupuy's projected time and Perfect Prediction, as discussed offline.}.

\subsection{Epistemic omniscience}

We assume epistemic omniscience for everything relevant to us. Epistemic omniscience applies in particular if $\mathcal{K}_i$ is the equality relation for all agents.

Epistemic omniscience is meant, here, as factual omniscience, i.e., it subsumes deductive omniscience\footnote{$KA\wedge K(A\Rightarrow B) \implies K(B)$}, weak logical omniscience\footnote{Truths of propositional logic are known} and strong logical omniscience\footnote{Truths of modal logic are known}: any proposition that holds is known to the agent.

First, we assume that strategies are always known in all possible worlds (this entails necessary knowledge of strategies in all normal worlds).

$$\forall i \in P, \forall w, w' \in \Omega, w \mathcal{K}_i w'  \implies \overrightarrow{\sigma}(w)=\overrightarrow{\sigma}(w')$$

Next, we assume that the logical accessibility of any world (or the lack thereof) is known in all possible worlds (which entails necessary knowledge of logical accessibility in all normal worlds):

$$\forall i \in P, \forall w, w', w'' \in \Omega, w \mathcal{K}_i w' \wedge w\mathcal{L}_i w'' \implies w' \mathcal{L}_i w''$$

Thus, we will omit the universal quantifier on epistemically possible worlds and use the actual world on the left-hand-side of $\mathcal{L}_i$. The condition above is actually entailed by the construction of $(\mathcal{L}_i)_i$.

\bigskip

Thirdly, we assume that the closest-state function, for any agent and deviation of strategy, sends all worlds in a given epistemic partition to worlds in a unique (same or other) epistemic partition, i.e., the consequences of any deviation of strategies are known.

$$\forall i, j \in P, \forall s'_j \in \Sigma_j, \forall w, w'\in \Omega, w\mathcal{K}_i w' \implies f(w, j, s'_j) \mathcal{K}_i f(w', j, s'_j)$$

As a consequence, we can write $\overrightarrow{\sigma}(f(w, i, s'_i))$ to represent the profile that would be reached if agent $i$ deviates to $s'_i$ without worrying about universal quantifiers on the epistemic partition.

Finally, the logical possibility of a deviation of strategies is known as well:

$$\forall i, j, k \in P, \forall s'_j \in \Sigma_j, \forall w, w'\in \Omega, w\mathcal{K}_i w' \wedge  w\mathcal{L}_k f(w, j, s'_j) \implies w\mathcal{L}_k f(w', j, s'_j)$$

We thus do not have to use a universal quantifier on epistemically accessible worlds to making the statement that a deviation is possible ($w\mathcal{L}_i f(w, i, s'_i)$). This holds if partitions of $(\mathcal{K}_i)_i$ do not cross the boundaries of $\Lambda$ and $\Xi$.

All in all, the consequence of epistemic omniscience is that the epistemic accessibility relation will appear rarely in our statements.

\subsection{Eventual logical omniscience}

We now introduce the tools characterizing levels of logical omniscience. Eventual logical omniscience, suggested here, provides a way to model the emergence of logical omniscience as agents progress in their logical reasonings.

We are looking at a class of Kripke models in which possible worlds are organized in layers representing various levels of logical omniscience. At the bottommost level (0), all worlds are logically impossible \citep{Rantala1982}. At level 1, a logically possible world is non-normal, meaning that nothing is logically necessary, and everything is logically possible \citep{Kripke1965}. At any upper level of logical omniscience, a logically possible world is normal, meaning that agents know which worlds are logically possible, and recursively know that a deviation of strategy would (counterfactually) lead to a world with one less level of logical omniscience, i.e., logical omniscience degrades counterfactually.

Note that, for any level, we may have logically impossible worlds in which the proposition of level-k logical omniscience is set manually as being true.

The actual world of a model will then be taken at a sufficiently high level of logical omniscience, and in particular, will be normal (level at least 2). However, we do not have full (counterfactual) logical omniscience in the sense that arbitrary counterfactual deviations would stay in normal worlds. In particular, rationality is necessary, although it is not necessary that rationality is necessary. We call this slightly weaker account for logical omniscience ``eventual logical omniscience''\footnote{This is a direct hint to the notion of eventual consistency in the database world.}.

Non-normal worlds \citep{Kripke1965} allow us to manage logical omniscience by bootstrapping the logical reasoning that leads to the outcomes that are actually reachable in logically possible worlds. In normal worlds, an agent's assessment of logically possible possible worlds matches the actually logically possible possible worlds. 

\begin{definition}[Level-k logical omniscience]
A logically possible world satisfies level-1 logical omniscience if it is non-normal ($\Lambda \setminus \Xi$).

For $k \ge 2$, a logically possible world $w\in \Lambda$ satisfies level-k logical omniscience if:
\begin{itemize}
\item it is normal ($w\in\Xi$), and
\item any deviation of strategies by any agent leads to a world (logically impossible\footnote{The truth assignment of the level is then set manually.} or not) satisfying level-$(k-1)$ logical omniscience.
\item there exists a world $w'$ satisfying level-$(k-1)$ logical omniscience, in which the same strategy profile is reached ($\overrightarrow{\sigma}(w)=\overrightarrow{\sigma}(w')$) and that is logically accessible from $w$ for all agents ($w\mathcal{L}_i w'$).\footnote{Since $k\ge 2$ and w is normal, this comes down to assuming that there is a logically possible world (in $\Lambda$ as well) with one less degree of logical omniscience and in which the same strategy profile is reached. This can be seen as a monotony property, i.e., if it is known that a strategy profile is logically impossible at level k of logical omniscience, then it is also known to be logically impossible at any higher level.}

For logically impossible worlds, the level of logical omniscience is assigned manually \citep{Rantala1982}.
\end{itemize}

The last condition says that the logical possibility of a strategy profile must propagate downward through logical omniscience levels, i.e., if a strategy profile was proven logically impossible at some level of logical omniscience, then it must also be impossible at higher levels of logical omniscience as well. In other words, the set of possible strategy profiles compatible with necessary rationality and necessary knowledge of strategies must decrease when logical omniscience increases.

\end{definition}

In particular, logical omniscience on or above level 2 in a logically possible world entails its normality.

\subsection{Rationality}

Rationality is defined as utility maximization. In this paper, rationality is assumed in the context of necessary knowledge of strategies.

Since we are assuming necessary knowledge of strategies, in the sense that strategies are known in all possible worlds, this allows us to define rationality in a very simple way. Indeed, logically accessibility is known, and the outcome of the Lewis closest state function is known as well for any agent and deviation of strategy.

\begin{definition}[Rationality]
A player is rational in a logically possible world $w\in \Lambda$ if no logically possible change of strategy would guarantee a higher payoff to them.

$$\forall s'_{i} \in \Sigma_i, w\mathcal{L}_if(w, i, {s'}_i) \implies u_i(\overrightarrow{\sigma}(w)) \ge  u_i(\overrightarrow{s}(f(w, i, {s'}_i)))$$

Note that the above definition  entails that a player is also rational if they believe that any other choice of strategy by them is logically impossible\footnote{meaning that it counterfactually leads to a logically impossible world}. This is actually the case for high enough levels of logical omniscience, as at most one logically possible profile remains.

\end{definition}

The truth value of rationality in logically impossible worlds is assigned manually \citep{Rantala1982}. In practice, it is always assigned to true, following \citet{Rantala1982}'s requirement that a proposition true in all logically possible worlds must always be assigned to true in logically impossible worlds\footnote{2.3.iii on page 109.}. More generally, any necessary proposition in normal worlds must be set to true in all logically impossible worlds, and its necessity to false.

\section{Modal logic}
\label{section-modal-logic}

We now extend our Kripke frames to models by adding truth assignments to propositions and sentences.

\subsection{Syntax}
\label{section-syntax}
We extend a Kripke frame $M$ to a model with a labelling function $\models$ that assign truth values to logical sentences, defined recursively as follows, for any world in $\Lambda$.

We define the following atomic formulas, for a world $w\in\Lambda$:

$$w \models \text{true, always}$$
$$w \models RAT \text{ if all agents are rational in w}$$
$$w \models K(\overrightarrow\sigma) \text{ if all agents know one another's strategies}$$
$$w \models OMN^k \text{ if there is level-k logical omniscience, for } k\ge 1$$
$$w \models \text{play}(\overrightarrow{s}) \text{ if the strategy profile picked in } w \text{ is } \overrightarrow{s}$$
$$w \models \text{play}_i(s_i) \text{ if the strategy picked by agent } i \text{ in } w \text{ is } s_i$$

We now introduce the connectives originally defined by \cite{Kripke1963}. 

We introduce the two following propositional logic connectives for any logically possible world $w\in\Lambda$:

$$w \models A \wedge B \text{ if } w \models A \text{ and } w \models B$$
$$w \models \neg A \text{ if } (M, w) \not \models A $$

As well as the derived connectives:

$$w \models A \vee B \text{ if } w \models \neg (\neg A \wedge \neg B)$$
$$w \models (A \Rightarrow B) \text{ if } w \models \neg A \vee B$$
$$w \models (A \iff B) \text{ if } w \models (A \Rightarrow B) \wedge (B \Rightarrow A)$$

For any normal world $w\in\Xi$, we also define the modal logic connectives:

$$w \models \Box A \text{ if A is necessary in w}$$

It is naturally extended to::

$$w \models \Diamond A \text{ if } w \models \neg \Box \neg A$$

Finally, we introduce the following operator for convenience\footnote{Note that, because of the way we defined levels of logical omniscience, the logical accessibility relation is the same for all agents in any normal world with level-k logical omniscience for $k\ge 1$: all worlds are accessible if $k=1$, and otherwise all worlds in $\Lambda$.}:

$$w \models \Diamond^c A \text{ if } \exists i,j \in P, s'_i \neq s_i(w), w\mathcal{L}_j f(w, i, s'_i) \wedge f(w, i, s'_i)  \models A$$

We can thus express in particular the logical possibility of a deviation of strategy, which is the logical possibility operator constrained to the range of the Lewis closest-state function)\footnote{$\Diamond^c$ actually implicitly involves a virtual accessibility relation that restricts the logical accessibility relation to one level below of logical omniscience plus the actual world. It is this same virtual accessibility relation that the semantics of our closest-state functions is based on.}:

$$w \models \Diamond^c \text{play}_i(s'_i) \text{ if } w\mathcal{L}_j f(w, i, s'_i) \text{ for some } j$$

\subsection{Logically impossible worlds}

For logically impossible worlds, we assign each proposional logic statement to its truth value manually \citep{Rantala1982}. For example, we may have for a non-normal world $w \models RAT$ independently of whether the definition of rationality applies. These worlds can be inconsistent regarding classical propositional logic, e.g., we may not have weak consistency. In this paper, the terminology ``logically possible'' refers to propositional logic.

\subsection{Non-normal worlds}

Logically possible non-normal worlds follow the laws of propositional logics, meaning that the truth values of propositions are set according to Section \ref{section-syntax}.

However, in all non-normal worlds (including logically impossible worlds), modal logic connectives are set according to different rules, as introduced by \citet{Kripke1965} in his seeding paper on impossible possible worlds.

$$w \models \Box A \text{ never holds for any event A }$$
$$w \models \Diamond A \text{ always holds for any event A}$$
$$w \models \Diamond^c A \text{ always holds for any event A}$$

In other worlds, in a non-normal world, nothing is necessary, and everything is possible.

\subsection{Full support by a Kripke structure}

We now characterize the Kripke structures that truly cover all possible agents' choices. This can be seen as a weak form of free will, i.e., the agents can pick any strategy they want and that they believe is logically possible.

Requiring full support is crucial, because unlike in Nashian settings such as rationalizability, or in Halpern's and Pass's minimax-rationalizability, the order of elimination of strategy profiles matters. Kripke models that are missing, for one reason or another, otherwise possible deviations of strategies may thus diverge to a different outcome.

\begin{definition}[Level-k full-support by a Kripke structure]
A Kripke structure $M=(\Omega, \Lambda, \Xi, \overrightarrow{\sigma}, \mathcal{K}, \mathcal{L}, \models)$, appropriate for a game $\Gamma$ and satisfying the assumptions stated previously, is said to have level-0 full support with respect to a logical formula $F$ if $\forall w\in\Omega, w\models F$ and $\forall w\in\Xi, w\models \Box F$ as well as

$$\forall \overrightarrow{s}\in\Sigma, \exists w\in\Lambda, w\models \text{play}(\overrightarrow{s})$$

$M$ has level-1 full support with respect to a logical formula $F$ if it has level-0 full support, and for any Kripke structure $N$  satisfying level-0 full support and logically possible\footnote{and non-normal because of the required antecedent} world $x\in\Lambda_N$,

$$x\models_N F \wedge OMN^1  \implies \exists w\in\Lambda, w\models \text{play}(\overrightarrow{\sigma_N}(x)) \wedge F \wedge OMN^1$$

For $k\ge 2$, $M$ has level-k full support with respect to a logical formula $F$ if it has level-$(k-1)$ full support, and for any Kripke structure $N$  satisfying level-$(k-1)$ full support and world $x\in\Lambda_N$,

$$x\models_N \Box F \wedge OMN^k  \implies \exists w\in\Lambda, w\models \text{play}(\overrightarrow{\sigma_N}(x)) \wedge \Box F \wedge OMN^k$$

A Kripke structure satisfying the above condition for all $k$ is said to have full-support with respect to F.

\end{definition}

We will in particular use the above class with (necessary) rationality and (necessary) knowledge of strategies, i.e., $F=RAT\wedge K(\overrightarrow{\sigma})$.

The condition then says that a full-support Kripke model cannot "miss" a strategy profile realized in a world at some level $k$ of logical omniscience under necessary rationality and necessary knowledge of strategies, if that strategy profile could be obtained in a logically possible world in another Kripke model that has full-support at level $k-1$. It can be seen as some kind of completeness feature\footnote{Not in the sense that completeness has in logic, which is why we are using a different terminology to avoid confusion.}, obtained in increasing levels of logical omniscience.

\section{Characterization of the Perfectly Transparent Equilibrium in Kripke semantics}
\label{section-characterization}
We now state how the PTE and the iterated deletion of preempted strategy profile relate to Kripke frames as defined in the previous sections.

\subsection{Theorems and characterization}

First, we start with the levels of iterated elimination.

At level 1, we characterize individual rationality with logically possible, non-normal worlds with rationality and knowledge of strategies. This result is well known in literature.\footnote{Meaning that individual rationality corresponds in Kripke semantics to rationality and knowledge of strategies. The non-normal aspect, meaning that all deviations of strategies are considered possible, is mentioned explicitly in this paper.}

\begin{theorem}
\label{lemma-epistemic-first-level}
Let $\Gamma=(P, \Sigma, u)$ be a game in normal form.

Then the following is equivalent:

\begin{enumerate}[label=(\alph*)]
\item $\overrightarrow{s} \in \mathcal{S}_1(\Gamma)$, that is, $\overrightarrow{s}$ is individually rational.
\item There exists a Kripke structure $M=(\Omega, \Lambda, \Xi, \overrightarrow{s}, \mathcal{K}, \mathcal{L}, f, \models)$ that has full support with respect to rationality and knowledge of strategies, and a logically possible, non-normal world $w\in\Lambda\setminus\Xi$ such that

$$w\models\text{play}(\overrightarrow{s}) \wedge RAT \wedge K(\overrightarrow{\sigma}) \wedge OMN^1$$
\end{enumerate}

In other words, the level-$k$ individually rational outcomes are exactly the outcomes $o$ for which there exists a full-support Kripke model in which the agents reach this outcome $o$ under rationality, knowledge of strategies and level-1 (non-normal) logical omniscience.

\end{theorem}

Next, we characterize each higher level of elimination of non-individually-rational strategy profiles with higher levels of logical omniscience in normal models.

\begin{theorem}
\label{lemma-epistemic-higher-level}
Let $\Gamma=(P, \Sigma, u)$ be a game in normal form and $k\ge 2$.

Then the following is equivalent:

\begin{enumerate}[label=(\alph*)]
\item $\overrightarrow{s} \in \mathcal{S}_k(\Gamma)$
\item There exists a Kripke structure $M=(\Omega, \Lambda, \Xi, \overrightarrow{s}, \mathcal{K}, \mathcal{L}, f, \models)$ that has full support with respect to rationality and knowledge of strategies, and a logically possible world $w\in\Lambda$ such that

$$w\models\text{play}(\overrightarrow{s}) \wedge \Box RAT \wedge \Box K(\overrightarrow{\sigma}) \wedge OMN^k$$
\end{enumerate}

In other words, the level-$k$ individually rational outcomes are exactly the outcomes $o$ for which there exists a full-support Kripke model in which the agents reach this outcome $o$ under necessary rationality, necessary knowledge of strategies and level-k logical omniscience.

\end{theorem}

Finally, we can characterize Perfectly Transparent Equilibria in terms of Kripke models by going to arbitrarily high levels of logical omniscience.

\begin{theorem}
\label{theorem-pte-kripke}
Let $\Gamma=(P, \Sigma, u)$ be a game in normal form.

Then the following is equivalent:

\begin{enumerate}[label=(\alph*)]
\item $\overrightarrow{s}$ is a PTE for $\Gamma$
\item For any $k\ge 2$, there exists a Kripke structure $M=(\Omega, \Lambda, \Xi, \overrightarrow{s}, \mathcal{K}, \mathcal{L}, f, \models)$ that has full support with respect to rationality and knowledge of strategies, and a logically possible world $w\in\Lambda$ such that

$$w\models\text{play}(\overrightarrow{s}) \wedge \Box RAT \wedge \Box K(\overrightarrow{\sigma}) \wedge OMN^k$$
\end{enumerate}

In other words, the PTE characterizes the outcomes $o$ for which there exists a full-support Kripke model in which the agents reach this outcome $o$ under necessary rationality, necessary knowledge of strategies and an arbitrarily high level of logical omniscience

\end{theorem}

\subsection{Canonical Kripke structure}

In order to prove these lemmas, we are introducing some tools, in particular, the Canonical Kripke structure, which serves as a constructive proof of existence of a counterfactual structure underlying the PTE.

\begin{definition}[Canonical Kripke structure]

For a game $\Gamma$, we define the canonical Kripke structure\footnote{The difference between a structure and a model is that, in a model, we specify an actual world.}

$$M(\Gamma)=(\Omega, \Lambda, \Xi, \overrightarrow{\sigma}, \mathcal{K}, \mathcal{L}, \models)$$

 recursively as follows.

Let $\Omega= \Sigma \times \mathbb{N}$. Worlds are thus organized in levels.

Let $\Lambda=\{ (\overrightarrow{s}, k) \in \Sigma \times \mathbb{N} | k \ge 1 \wedge \overrightarrow{s}\in\mathcal{S}_k(\Gamma) \}$. In particular, worlds on level 0 are thus logically impossible worlds.

Let $\Xi=\{ (\overrightarrow{s}, k) \in \Sigma \times \mathbb{N} | k \ge 2 \wedge \overrightarrow{s}\in\mathcal{S}_k(\Gamma) \}$. In particular, worlds on level 1 (and 0) are thus non-normal worlds.

Let $\overrightarrow{\sigma}$ be defined as

$$\forall w=(\overrightarrow{s}, k)\in \Omega, \overrightarrow{\sigma}(w)=\overrightarrow{s}$$

We also introduce the "level" notation:

$$\forall w=(\overrightarrow{s}, k)\in \Omega, \Lambda, \lambda(w)=k$$

Let $(\mathcal{K}_i)_i$ be defined as the equality relation, i.e.:

$$\forall i\in P, \forall w, w'\in \Omega, w\mathcal{K}_i w' \iff w=w'$$

Let $(\mathcal{L}_i)_i$ be defined as:

$$\forall i\in P, \forall w, w'\in \Omega, w\mathcal{L}_i w' \iff w' \in 	\Lambda \vee w\notin \Xi$$

Note that this is a total relation on $\Lambda$, and any world is logically accessible from $\Omega\setminus\Xi$.

Let $f$ be defined as a "worst case scenario" w.r.t. the opponent's reaction, that is (the first case that applies is taken):

\begin{align}
\forall w \in \Omega, \forall i\in P, \forall s'_i\in\Sigma_i,f(w, i, s'_i)=\\
\begin{cases}
((s'_i, s_{-i}), 0)& \text{ if } \lambda(w)=0\\
w & \text{ if } \sigma_i(w)=s'_i\\
((s'_i, \displaystyle \arg \min_{
\begin{array}{c}
\tau_{-i}\in\Sigma_{-i}
\\
{\scriptstyle \text{s.t.} (s'_i, \tau_{-i})\in \mathcal{S}_{\lambda(w)-1}(\Gamma)}
\end{array}
} u_i(s'_i, \tau_{-i}) ), \lambda(w)-1)& \text{ if the min does not diverge }\\
((s'_i, s_{-i}), \lambda(w)-1)& \text{ otherwise }\\
\end{cases}
\end{align}

For any logically impossible world $w\notin\Lambda$, we set

$$w \models RAT$$
$$w \models K(\overrightarrow\sigma)$$
$$w \models \text{play}(\overrightarrow\sigma(w))$$
$$w \models OMN^{\lambda(w)}$$

as well as the transitive closure thereof via $\wedge$ and $\neg$.

For non-normal worlds, logically possible or impossible, we also set modal operators in such a way that anything is possible and nothing is necessary.

\end{definition}

\subsection{Properties of canonical Kripke structures}

The canonical Kripke structure has a few properties that we express as lemmas.

\begin{lemma}[Cascading levels of counterfactuals in canonical Kripke structures]
\label{lemma-cascading-levels}

Let $\Gamma$ be a game and $M(\Gamma)=(\Omega, \Lambda, \Xi, \overrightarrow{\sigma}, \mathcal{K}, \mathcal{L})$ its canonical Kripke structure.

We have, for any $w\in\Omega$ with $\lambda(w)\ge1$ and any $s'_i\neq \sigma_i(w)$: 

$$\lambda(f(w, i, s'_i)) = \lambda(w)-1$$

and for $w\in\Lambda$ with $\lambda(w)\ge1$ and any $s'_i\neq \sigma_i(w)$: 

$$w\models \Diamond^c \text{play}_i(s'_i) \iff \overrightarrow{s}(f(w, i, s'_i))\in\mathcal{S}_{\lambda(w)-1}$$

\end{lemma}

\begin{proof}[Lemma \ref{lemma-cascading-levels}]
The first equality is obtained for any level $k \ge 1$, because $f(w, i, s'_i)$ is one level below $w$ for any $s'_i\neq \sigma_i(w)$.

For the second statement, left to right, we have $w\mathcal{L}_i f(w, i, s'_i)$ at level 1 because w is non-normal. The set membership is trivial for $k=1$ as $\mathcal{S}_0=\Sigma$. For higher levels, we have $f(w, i, s'_i)\in\Lambda$ and the set membership holds by definition of $\Lambda$.

For the second statement, right to left, the left-hand side holds trivially for $\lambda(w)=1$ by definition of non-normality. For higher levels, if $\overrightarrow{s}(f(w, i, s'_i))\in\mathcal{S}_{\lambda(w)-1}$, then the minimum in the definition of $f$ does not diverge ($s'_i$ is in the projection  of $\mathcal{S}_{\lambda(w)-1}$ onto agent $i$'s strategy space), and thus $f(w, i, s'_i)\in\Lambda$, $w\mathcal{L}_i f(w, i, s'_i)$ and $w\models \Diamond^c \text{play}_i(s'_i)$.

\end{proof}

\begin{lemma}[Level-k logical omniscience in canonical Kripke structures]
\label{lemma-level-k-logical-omniscience}

Let $\Gamma$ be a game and $M(\Gamma)=(\Omega, \Lambda, \Xi, \overrightarrow{\sigma}, \mathcal{K}, \mathcal{L})$ its canonical Kripke structure.

We have, for any $w\in\Omega$ such that $\lambda(w)\ge 1$,

$$\forall w\in \Omega, w \models OMN^{\lambda(w)}$$

\end{lemma}

\begin{proof}[Lemma \ref{lemma-level-k-logical-omniscience}]

First, for logically impossible worlds, this is true because we set the truth assignment so.

Let us start with a logically possible, non-normal world w such that $$\lambda(w)=1$$

On level 1, this hold simply because w is non-normal, which characterizes level-1 logical omniscience:

$$w \models OMN^1$$

Now, let us assume that this is true for $\lambda(w)=k\ge 1$ and prove that this holds for $\lambda(w)=k+1$.

Let us pick any logically possible (and normal) $w$ such that $$\lambda(w)=k+1\ge 2$$

First, w is normal.

Second, for any (logically possible or not) deviation of strategy, we know from Lemma \ref{lemma-cascading-levels} that:

$$\lambda(f(w, i, s'_i))=k$$

Using the induction hypothesis, we have:

$$f(w, i, s'_i) \models OMN^k$$

Finally, $(\overrightarrow\sigma(w), k)$ is a world satisfying level-k logical omniscience in which the same strategy profile is reached. Because $\mathcal{S}_{k+1}\subset\mathcal{S}_k$, $\overrightarrow\sigma(w)\in\mathcal{S}_k$, $(\overrightarrow\sigma(w), k)\in\Lambda$, and it is thus logically accessible from w.

The above three points fulfill the definition of level-k logical omniscience and thus:

$$w \models OMN^{k+1}$$

This ends the recursion. \qed

\end{proof}

\begin{lemma}[Necessary rationality in canonical Kripke structures]
\label{lemma-necessary-rationality}

Let $\Gamma$ be a game and $M(\Gamma)=(\Omega, \Lambda, \Xi, \overrightarrow{\sigma}, \mathcal{K}, \mathcal{L}, \models)$ its canonical Kripke structure.

We have, for any $w\in\Omega$ (and in particular $w\in\Lambda$):

$$w \models RAT$$

And thus, for any $w\in\Xi$:

$$w \models \Box RAT$$

\end{lemma}

\begin{proof}[Lemma \ref{lemma-necessary-rationality}]

We have to prove rationality in any possible world.

For logically impossible worlds $w\notin\Lambda$, rationality holds because the truth assignment is set that way 
 ($w \models RAT$).

Let $w\in\Lambda$ be a logically possible world (thus on a level at least 1). We prove that all agents are rational in this world.

We have $\lambda(w)=k\ge 1$.

Let $i\in P$ be an agent. If no deviation is logically possible for agent $i$, then we have already established that the agent is rational in $w$ by definition of rationality.

If there is one possible deviation for agent $i$, let us assume that $s'_i$ is a logically possible deviation of strategy

$$w\models \Diamond^c \text{play}_i(s'_i)$$

and show that this deviation would give $i$ a worse or equal payoff.

Since $w\in\Lambda$, we have $\overrightarrow{s}=\overrightarrow{\sigma}(w)\in\mathcal{S}_{\lambda(w)}$, and we have by definition of $(\mathcal{S}_k)_{k\in \mathbb{N}}$:

\begin{align}
 u_i(\overrightarrow{s}) \ge \max_{
\begin{array}{c}
\tau_i\in\Sigma_i
\\
{\scriptstyle \text{s.t.} \exists \tau_{-i}\in\Sigma_{-i}, (\tau_i, \tau_{-i})\in \mathcal{S}_{k-1}}
\end{array}
}
\quad
\min_{
\begin{array}{c}
\tau_{-i}\in\Sigma_{-i}
\\
{\scriptstyle \text{s.t.} (\tau_i, \tau_{-i})\in \mathcal{S}_{k-1}}
\end{array}
} u_i(\tau_i, \tau_{-i})
\end{align}

Since $s'_i$ is logically possible, $f(w', i, s'_i)\in\Lambda$, and we have because of Lemma \ref{lemma-cascading-levels}:

$$\overrightarrow{s}(f(w', i, s'_i))\in \mathcal{S}_{k-1}$$

And thus:

$$\exists \tau_{-i}\in\Sigma_{-i}, (s'_i, \tau_{-i})\in \mathcal{S}_{k-1}(\Gamma)$$

Thus $s'_i$ belongs to the set over which the maximum is taken, and thus by definition of the maximum:

\begin{align}
w\models \Diamond^c \text{play}_i(s'_i) \implies   u_i(\overrightarrow{s}) \ge 
\min_{
\begin{array}{c}
\tau_{-i}\in\Sigma_{-i}
\\
{\scriptstyle \text{s.t.} (s'_i, \tau_{-i})\in \mathcal{S}_{k-1}(\Gamma)}
\end{array}
} u_i(s'_i, \tau_{-i})
\end{align}

And by definition of the argmin:

\begin{align}
 w\models \Diamond^c \text{play}_i(s'_i) \implies  u_i(\overrightarrow{s}) \ge  u_i(s'_i, \arg \min_{
\begin{array}{c}
\tau_{-i}\in\Sigma_{-i} \\
{\scriptstyle \text{s.t.} (s'_i, \tau_{-i})\in \mathcal{S}_{k-1}(\Gamma)}
\end{array}} u_i(s'_i, \tau_{-i}))
\end{align}

By definition of $f$ for a logically possible deviation at level at least 1:

\begin{align}
w\models \Diamond^c \text{play}_i(s'_i) \implies   u_i(\overrightarrow{s}) \ge  u_i(\overrightarrow{\sigma}(f(w', i, s'_i)))
\end{align}

This holds for any $i$, $s'_i$, thus rationality holds at $w$:

$$w \models RAT$$

Thus,

$$\forall w\in\Omega,  w \models RAT$$

And thus, since all agents are rational in all logically possible worlds, by definition of the necessity operator in a normal world, where exactly the logically possible worlds are logically accessible:

$$\forall w\in\Xi, w \models \Box RAT$$

This finishes the proof. \qed

\end{proof}

\begin{lemma}[Necessary knowledge of strategies in canonical Kripke structures]
\label{lemma-necessary-knowledge-of-strategies}

Let $\Gamma$ be a game and $M=(\Omega, \Lambda, \Xi, \overrightarrow{\sigma}, \mathcal{K}, \mathcal{L}, \models$ its canonical Kripke structure.

We have, for any $w\in\Omega$ and in particular $w\in\Lambda$: 

$$w \models K(\overrightarrow{\sigma})$$

And thus for any $w\in\Xi$:

$$w \models \Box K(\overrightarrow{\sigma})$$

\end{lemma}

\begin{proof}[Lemma \ref{lemma-necessary-knowledge-of-strategies}]

This follows directly from the fact that the epistemic accessibility relation is the equality relation, and from the definition of $\Box$ in normal worlds.

\end{proof}

\begin{lemma}[Agent decisions in a canonical Kripke structure]
\label{lemma-agent-decisions}

Let $\Gamma$ be a game and $M(\Gamma)=(\Omega, \Lambda, \Xi, \overrightarrow{\sigma}, \mathcal{K}, \mathcal{L}, \models)$ its canonical Kripke structure.

We have, for any $w\in\Omega$: 

$$w \models \text{play}(\overrightarrow{\sigma}(w))$$

\end{lemma}

\begin{proof}[Lemma \ref{lemma-necessary-knowledge-of-strategies}]

This follows directly from the definition of $\overrightarrow\sigma$ for logically possible worlds. For impossible worlds, the truth assignment is set that way.

\end{proof}

\begin{lemma}[Full-support by the canonical Kripke structure with respect to rationality and knowledge of strategies]
\label{lemma-completeness}

Let $\Gamma$ be a game and $M(\Gamma)=(\Omega, \Lambda, \Xi, \overrightarrow{\sigma}, \mathcal{K}, \mathcal{L}, \models)$ its canonical Kripke structure.

$M$ has full support with respect to $RAT\wedge K(\overrightarrow\sigma)$.
\end{lemma}

\begin{proof}[Lemma \ref{lemma-completeness}]

We start with level 0.

We know that $\forall w\in\Omega, w\models RAT$ and $\forall w\in\Xi, w\models \Box RAT$ from Lemma \ref{lemma-necessary-rationality}.

We know that $\forall w\in\Omega, w\models K(\overrightarrow\sigma)$ and $\forall w\in\Xi, w\models \Box K(\overrightarrow\sigma)$ from Lemma \ref{lemma-necessary-knowledge-of-strategies}.

Thus $\forall w\in\Omega, w\models RAT \wedge K(\overrightarrow\sigma)$ and $\forall w\in\Xi, w\models \Box (RAT \wedge K(\overrightarrow\sigma))$.

Furthermore, we have by construction, for each strategy profile, a world in which that strategy profile is played.

The conditions of level-0 full-support are thus fulfilled.

We know prove the result for $k=1$.

Let $N$ be a Kripke structure with level-1 full support. Let us assume there is a logically possible, non-normal world $x\in\Lambda_N$ such that 
$$x \models_N RAT  \wedge K(\overrightarrow{\sigma}_N)  \wedge OMN^1$$

We need to find an $w\in\Lambda$ such that

$$w\models \text{play}(\overrightarrow{\sigma_N}(x)) \wedge RAT  \wedge K(\overrightarrow{\sigma}) \wedge OMN^1$$

We first prove that $\overrightarrow{\sigma_N}(x)\in\mathcal{S}_1$, i.e., $\overrightarrow\sigma_N(x)$ is individually rational. This result is known and commonly found in classical game theory literature.

Because of rationality in $x$, we have, for any deviation\footnote{x is non-normal, so all deviations are logically possible in x} of strategies $s'_i$:

$$u_i(\overrightarrow{\sigma}_N(x)) \ge u_i(\overrightarrow{\sigma_N}(f_N(x, i, s'_i)))$$

We also have by definition of the minimum:

$$u_i(\overrightarrow{\sigma_N}(f_N(x, i, s'_i))) \ge\min_{
\begin{array}{c}
\tau_{-i}\in\Sigma_{-i}
\end{array}
} u_i(s'_i, \tau_{-i})$$

Thus, for any deviation of strategies $s'_i$ by agent $i$ in world $x$:

$$u_i(\overrightarrow\sigma_N(x)) \ge \min_{
\begin{array}{c}
\tau_{-i}\in\Sigma_{-i}
\end{array}
} u_i(s'_i, \tau_{-i})$$

The above also applies without deviation of strategy also by definition of the min:

$$u_i(\overrightarrow\sigma_N(x)) \ge \min_{
\begin{array}{c}
\tau_{-i}\in\Sigma_{-i}
\end{array}
} u_i(\sigma_{N,i}(x), \tau_{-i})$$

And thus, taking the max:

$$\forall i \in P, u_i(\overrightarrow{\sigma}_N(x)) \ge \max_{
\begin{array}{c}
\tau_i\in\Sigma_i
\end{array}
}
\quad \min_{
\begin{array}{c}
\tau_{-i}\in\Sigma_{-i}
\end{array}
} u_i(\tau_i, \tau_{-i})$$

And thus

$$\overrightarrow{\sigma}_N(x)\in\mathcal{S}_1$$

We can then take $w=(\overrightarrow{\sigma_N}(x), 1)$. Then by using Lemmas \ref{lemma-necessary-rationality}, \ref{lemma-necessary-knowledge-of-strategies}, \ref{lemma-agent-decisions} and \ref{lemma-level-k-logical-omniscience}:

$$\exists w\in\Lambda, w\models \text{play}(\overrightarrow{\sigma_N}(x)) \wedge RAT  \wedge K(\overrightarrow{\sigma}) \wedge OMN^1$$

which proves level-1 full support.

We now prove the result recursively for $k+1\ge 2$. We assume that the lemma is true for $k$ (thus $k\ge 1$), i.e., for any Kripke structure $N$ with level with level-$(k-1)$ full support, we have

$$x\models_N \Box RAT  \wedge \Box K(\overrightarrow{\sigma})  \wedge OMN^k  \implies \exists w\in\Lambda, w\models \text{play}(\overrightarrow{\sigma_N}(x)) \wedge\Box RAT  \wedge \Box K(\overrightarrow{\sigma}) \wedge OMN^k$$

except in the special case that $k= 1$:

$$x\models_N RAT  \wedge K(\overrightarrow{\sigma})  \wedge OMN^1  \implies \exists w\in\Lambda, w\models \text{play}(\overrightarrow{\sigma_N}(x)) \wedge RAT  \wedge K(\overrightarrow{\sigma}) \wedge OMN^1$$

Let $N$ be a Kripke structure with level-k full support. Let us assume there is a logically possible world $x\in\Lambda_N$ such that 
$$x \models_N \Box RAT  \wedge \Box K(\overrightarrow{\sigma})  \wedge OMN^{k+1}$$

We need to find an $w\in\Lambda$ such that

$$w\models \text{play}(\overrightarrow{\sigma_N}(x)) \wedge \Box RAT  \wedge \Box K(\overrightarrow{\sigma}) \wedge OMN^{k+1}$$

We first prove that $\overrightarrow{\sigma_N}(x)\in\mathcal{S}_{k+1}$.

Because of rationality in $x$, we have, for any logically possible deviation of strategies $s'_i$:

$$u_i(\overrightarrow{\sigma}_N(x)) \ge u_i(\overrightarrow{\sigma_N}(f_N(x, i, s'_i))$$

By definition of level-$(k+1)$ logical omniscience:

$$f_N(x, i, s'_i)\models OMN^k$$

If $k\ge 2$, $f_N(x, i, s'_i)$ is normal and thus logically necessity applies here in the same way as in $x$:

$$f_N(x, i, s'_i)\models \Box RAT  \wedge \Box K(\overrightarrow{\sigma}) \wedge OMN^k$$

If $k=1$, it is non-normal and:

$$f_N(x, i, s'_i)\models RAT  \wedge K(\overrightarrow{\sigma}) \wedge OMN^1$$

Since we know that $M$ has level-k full support, we know that, in the case that $k\ge2$: 

$$\exists w\in\Lambda, w\models \text{play}(\overrightarrow{\sigma_N}(f_N(x, i, s'_i))) \wedge\Box RAT  \wedge \Box K(\overrightarrow{\sigma}) \wedge OMN^k$$

or if $k=1$:
$$\exists w\in\Lambda, w\models \text{play}(\overrightarrow{\sigma_N}(f_N(x, i, s'_i))) \wedge RAT  \wedge K(\overrightarrow{\sigma}) \wedge OMN^1$$

By definition of $M$, because these worlds are logically possible, we know that thus $\overrightarrow{\sigma_N}(f(x, i, s'_i))\in\mathcal{S}_k$ for each possible deviation $s'_i$.

Thus by definition of the min and by transitivity:

$$u_i(\overrightarrow\sigma_N(x)) \ge \min_{
\begin{array}{c}
\tau_{-i}\in\Sigma_{-i}
\\
{\scriptstyle \text{s.t.} (s'_i, \tau_{-i})\in \mathcal{S}_k}
\end{array}
} u_i(s'_i, \tau_{-i})$$
a

Furthermore, because $N$ has full level-k support, and $M$ has logically possible worlds with level-k logical omniscience for each strategy profile in $\mathcal{S}_k$, we know that the possible deviations of strategies for agent $i$ in world $x$ are thus characterized by the projection of $\mathcal{S}_k$ on $i$'s strategy space, with the exception of $s_i=\sigma_{N,i}(x)$ (no change of strategy does not constitute a deviation). We are thus missing one more inequality to get the complete set of logically possible strategies for agent $i$, namely, that the actual payoff of agent $i$ also dominates the minimum payoff remaining in $\mathcal{S}_k$ for his actual strategy:

$$u_i(\overrightarrow{\sigma}_N(x)) \ge \min_{
\begin{array}{c}
\tau_{-i}\in\Sigma_{-i}
\\
{\scriptstyle \text{s.t.} (s_i, \tau_{-i})\in \mathcal{S}_k}
\end{array}
} u_i(s_i, \tau_{-i})$$

However, by definition of level-{k+1} logical omniscience, we do know that there is a logically possible world in $N$ with level-k logical omniscience, in which $\overrightarrow{\sigma}_N(x)$ is played. It follows from M having level-k full support that $\overrightarrow{\sigma}_N(x)\in\mathcal{S}_k$. Thus, the equality above holds as well.

Thus, it follows that the following holds for any $s'_i$, and not only deviations, in the projection of $\mathcal{S}_{k-1}$ on agent $i$'s strategy space.

$$u_i(\overrightarrow{\sigma}_N(x)) \ge \min_{
\begin{array}{c}
\tau_{-i}\in\Sigma_{-i}
\\
{\scriptstyle \text{s.t.} (s'_i, \tau_{-i})\in \mathcal{S}_k(\Gamma)}
\end{array}
} u_i(s'_i, \tau_{-i})$$

Thus, taking the max:

$$\forall i \in P, u_i(\overrightarrow{\sigma}_N(x)) \ge \max_{
\begin{array}{c}
\tau_i\in\Sigma_i
\\
{\scriptstyle \text{s.t.} \exists \tau_{-i}\in\Sigma_{-i}, (\tau_i, \tau_{-i})\in \mathcal{S}_k(\Gamma)}
\end{array}
}
\quad \min_{
\begin{array}{c}
\tau_{-i}\in\Sigma_{-i}
\\
{\scriptstyle \text{s.t.} (\tau_i, \tau_{-i})\in \mathcal{S}_k(\Gamma)}
\end{array}
} u_i(\tau_i, \tau_{-i})$$

And thus

$$\overrightarrow{\sigma}_N(x)\in\mathcal{S}_{k+1}$$

We can then take $w=(\overrightarrow{\sigma_N}(x), k+1)$. Then by using Lemmas \ref{lemma-necessary-rationality}, \ref{lemma-necessary-knowledge-of-strategies}, \ref{lemma-agent-decisions} and \ref{lemma-level-k-logical-omniscience}:

$$\exists w\in\Lambda, w\models \text{play}(\overrightarrow{\sigma_N}(x)) \wedge\Box RAT  \wedge \Box K(\overrightarrow{\sigma}) \wedge OMN^{k+1}$$

Which finishes the proof. \qed

\end{proof}

\subsection{Proof of the main theorems}
With the above lemmas, we are in a position to prove our main result.

\begin{proof}[Theorem \ref{lemma-epistemic-first-level}]

We first prove that (a) implies (b) where, in (b), the canonical Kripke structure is meant as $M$.

Let $\overrightarrow{s} \in \mathcal{S}_1(\Gamma)$ by a level-k individually rational outcome.

Let $w=(\overrightarrow{s}, 1)\in\Omega$. This is a logically possible, non-normal world by definition of $\Lambda$.

Firstly, Lemma \ref{lemma-agent-decisions} gives us:

\begin{align}
\label{equ-play-h}
w \models \text{play}(\overrightarrow{s})
\end{align}

Secondly, Lemma \ref{lemma-level-k-logical-omniscience} gives us:

$$\forall w\in \Omega, w \models OMN^1$$

Thirdly, Lemma \ref{lemma-necessary-rationality} gives us:

$$w \models RAT$$

Fourthly, Lemma \ref{lemma-necessary-knowledge-of-strategies} gives us:

$$w \models K(\overrightarrow{\sigma})$$

Because $w$ is logically possible, we thus can establish that $w$ fulfils the required condition:

$$(M,w) \models \text{play}(\overrightarrow{s}) \wedge RAT\wedge K(\overrightarrow{\sigma}) \wedge OMN^1 $$

\bigskip

We now prove that (b) implies (a).

Let $N$ be a complete Kripke model with full support and $x\in\Lambda_N$ such that 

$$x \models_N \text{play}(\overrightarrow{s}) \wedge RAT\wedge K(\overrightarrow{\sigma}) \wedge OMN^1 $$

We know that the canonical Kripke structure has full support. Thus, by definition of level-1 full support there must exist a world $w$ fulfilling

$$w \models \text{play}(\overrightarrow{s}) \wedge  RAT\wedge K(\overrightarrow{\sigma}) \wedge OMN^1 $$

By definition of $\Lambda$ and because of Lemma \ref{lemma-cascading-levels}:

$$\overrightarrow{\sigma}(w)=\overrightarrow{s}\in\mathcal{S}_{\lambda(w)}=\mathcal{S}_1$$

Which finishes the proof. \qed

\end{proof}
\begin{proof}[Theorem \ref{lemma-epistemic-higher-level}]

We first prove that (a) implies (b) where, in (b), the canonical Kripke structure is meant as $M$.

Let $\overrightarrow{s} \in \mathcal{S}_k(\Gamma)$ by a level-k individually rational outcome for some $k\ge 1$.

Let $w=(\overrightarrow{s}, k)\in\Omega$. This is a logically possible world by definition of $\Lambda$.

Firstly, Lemma \ref{lemma-agent-decisions} gives us:

\begin{align}
\label{equ-play-h}
w \models \text{play}(\overrightarrow{s})
\end{align}

Secondly, Lemma \ref{lemma-level-k-logical-omniscience} gives us:

$$\forall w\in \Omega, w \models OMN^k$$

Thirdly, Lemma \ref{lemma-necessary-rationality} gives us:

$$w \models \Box RAT$$

Fourthly, Lemma \ref{lemma-necessary-knowledge-of-strategies} gives us:

$$w \models \Box K(\overrightarrow{\sigma})$$

Because $w$ is logically possible, we thus can establish that $w$ fulfils the required condition:

$$(M,w) \models \text{play}(\overrightarrow{s}) \wedge \Box RAT\wedge \Box K(\overrightarrow{\sigma}) \wedge OMN^k $$

\bigskip

We now prove that (b) implies (a).

Let $N$ be a complete Kripke model with full support and $x\in\Lambda_N$ such that 

$$x \models_N \text{play}(\overrightarrow{s}) \wedge \Box RAT\wedge \Box K(\overrightarrow{\sigma}) \wedge OMN^k $$

We know that the canonical Kripke structure has full support. Thus, by definition of level-k full support there must exist a world $w$ fulfilling

$$w \models \text{play}(\overrightarrow{s}) \wedge \Box RAT\wedge \Box K(\overrightarrow{\sigma}) \wedge OMN^k $$

By definition of $\Lambda$ and because of Lemma \ref{lemma-cascading-levels}:

$$\overrightarrow{\sigma}(w)=\overrightarrow{s}\in\mathcal{S}_{\lambda(w)}=\mathcal{S}_k$$

Which finishes the proof. \qed

\end{proof}

\begin{proof}[Theorem \ref{theorem-pte-kripke}]

If $\overrightarrow{s}$ is a PTE, then it belongs to $\mathcal{S}_k$ for any $k\ge 1$. It follows from Theorem \ref{lemma-epistemic-higher-level} that there is a Kripke model (structure and actual world) fulfilling the conditions for any value of $k$.

Conversely, if we can find a Kripke model fulfilling these conditions for any value of $k\ge 1$, then $\forall k\ge 1, \overrightarrow{s}\in \mathcal{S}_k$ and thus $\overrightarrow{s}\in\cap_{k\ge 1} \mathcal{S}_k$, which characterizes a PTE.

\end{proof}

\section{Conclusion and summary}
\label{section-conclusion}

In this paper, we characterized Perfectly Transparent Equilibria for strategic games with no ties in terms of adapted Kripke semantics.

From an algorithmic perspective, a PTE is obtained by iterated deletion of non-individually rational strategy profiles.

From a Kripke semantics perspective, a PTE is characterized by necessary rationality, epistemic omniscience (in particular necessary knowledge of strategies) and eventual logical omniscience.

The need to weaken the assumption of full logical omniscience suggest that an impossible triangle may exist.

\begin{conjecture}[Impossibility triangle]

There cannot be ``useful'' Kripke models having simultaneously:
\begin{enumerate}
\item Epistemic (factual) omniscience
\item Full logical omniscience
\item Necessary rationality
\end{enumerate}

\end{conjecture}

We also gained insight at how the PTE differs from \citet{Halpern:2013aa}'s Common Counterfactual Belief of Rationality, in particular:
\begin{itemize}
\item CCBR deals with situations that are on a spectrum between a fully opaque setting (unilateral deviations) and a fully transparent setting (epistemic omniscience), but excluding fully transparent settings. Halpern and Pass coined the term ``translucency''. The PTE accounts for the end of the spectrum where decisions are fully transparent (necessary knowledge of strategies). The existence of counterexamples showing that CCBR and the PTE may diverge demonstrates the existence of a singularity at that end of the spectrum.
\item CCBR involves probability distribution modeling beliefs at accessible worlds. The PTE uses no probabilities, as it has full epistemic omniscience and has a deterministic, albeit non-trivial, nature. The singularity is due to the elimination of some worlds proven as logically impossible.
\item CCBR recursively assumes that the \emph{other} agents are (counterfactually) rational in case of a deviation. The PTE relies on what Halpern and Pass name CB*RAT, which is stronger and more transparent.
\end{itemize}

Since the Hofstadter equilibrium reached by superrational agents on symmetric games is a special case of PTE, it follows that we have also contributed a formal Kripke semantics account for Douglas \citet{Hofstadter1983}'s work.

\section{Acknowledgements}

The idea underlying necessary rationality and necessary knowledge of strategies, also called perfect prediction in algorithmic descriptions of the PTE and PPE, is to be credited to Jean-Pierre Dupuy, who calls this feature essential prediction and saw the link between Newcomb's problem and the Prisoner's dilemma in the 1990s. Dupuy's showed in several papers that this is a viable, rational approach to achieving perfect transparency in decision settings. I am also thankful to St\'ephane Reiche, with whom we collaborated on the formalism of the PPE in extensive form. Discussions with Bernard Walliser and Joe Halpern also gave me very helpful pointers.

\bibliographystyle{spbasic}      

\end{document}